\documentclass[letterpaper, 10 pt, conference]{ieeeconf}  

\IEEEoverridecommandlockouts                              
\usepackage{graphics} 
\usepackage{epsfig} 
\usepackage{amsmath} 
\usepackage{amssymb}  

\usepackage{tikz}
\usepackage{eso-pic}
\usepackage{xcolor}

\renewcommand\fbox{\fcolorbox{black}{gray!10}}

\newcommand\acceptedmanuscripttext{%
	
	\textbf{\textcopyright 2026 European Control Association.}\\
	
	\textbf{This paper has been accepted for presentation at the \textit{2026 European Control Conference (ECC)}}.
}

\AddToShipoutPicture*{%
	\AtPageUpperLeft{%
		\begin{tikzpicture}[remember picture,overlay]
			\node[
			anchor=north,
			yshift=-10pt
			] at (current page.north)
			{%
				\fbox{%
					\parbox{\dimexpr\textwidth-2\fboxsep-2\fboxrule}{%
						\acceptedmanuscripttext
					}%
				}%
			};
		\end{tikzpicture}
	}%
}

\title{\LARGE \bf
	A Floquet Mode LQR for Orbital Station-Keeping in Cislunar Space
}

\author{António Nunes, Sérgio Brás, and Pedro Batista
	\thanks{This work was supported by LARSyS FCT funding (DOI: 10.54499/LA/P/0083/2020, 10.54499/UIDP/50009/2020, and 10.54499/UIDB/50009/2020), and by NEURASPACE Project, Contract No. 9, under Regulation (EU) 2021/241 of the European Parliament and of the Council of February 12, 2021 and the Portuguese Recovery and Resilience Program (PRR), in component 05-Capitalization and Business Innovation, under Notice No. 425 01/C05-i01/2021 of the Regulation of Mobilizing Agendas/Alliances for reindustrialization}
	\thanks{António W. Nunes is with the Institute for Systems and Robotics, Instituto Superior Técnico, Universidade de Lisboa, Portugal (e-mail: {\tt \footnotesize antonio.w.nunes@tecnico.ulisboa.pt}) }
	\thanks{Pedro Batista is with the Institute for Systems and Robotics,
		Instituto Superior Técnico, Universidade de Lisboa, Portugal}%
	\thanks{Sérgio Brás is with the European Space Agency, The Netherlands}%
}

\newtheorem{proposition}{Proposition}[section]

\newcommand{\norm}[1]{|\!|#1|\!|}

\begin{document}

\maketitle
\thispagestyle{empty}
\pagestyle{empty}

\begin{abstract}
	
A linear optimal control law for orbital station-keeping in the Earth-Moon Restricted Three Body Problem (R3BP) is developed via Linear Quadratic Regulator (LQR) theory. First, the cost function is established considering a periodic state-weight matrix, leveraging stability information of the target orbits retrieved through Floquet theory. Then, the resulting periodic Riccati differential equation is solved and local asymptotic stability guarantees are shown. Finally, the performance of the proposed LQR when tracking periodic orbits in the circular and elliptic R3BPs is analyzed numerically.
	
\end{abstract}

\section{Introduction}
	
NASA's Artemis program and the Gateway lunar orbital platform \cite{creech2022Artemis} have motivated a new wave of scientific research on the operation of spacecraft in cislunar space. At the academic level, the Restricted Three-Body Problem (R3BP) provides an adequate first approximation for spacecraft dynamics in this context. Current work is focused on tracking periodic orbits about the five Lagrange (equilibrium) points of the R3BP, due to their advantageous positioning and communications clearance. However, these solutions are unstable, meaning that a control strategy is imperative to guarantee orbit maintenance and accurate trajectory tracking. Typically, the control strategies are designed to drive deviations from a target point evolving along the nominal orbit to zero, which formally constitutes an orbital \textit{station-keeping} problem.

Traditionally, the orbital station-keeping problem has been primarily tackled through discrete solutions that resort to impulsive maneuvers. In particular, if the equations of motion (EoM) are linearized, the design of optimal controllers becomes relatively straightforward, leading to famous strategies such as differential correction methods \cite{folta2010stationkeeping, pavlak2012optimal}. However, the demanding maneuvers requested by discrete approaches rely on classic high-thrust propulsion systems, precluding the use of emerging electric propulsion technologies that are highly efficient but constrained to low-magnitude thrust. As an alternative, continuous station-keeping solutions, compatible with electric propulsion, have also been investigated. Yet, optimal continuous strategies often lead to large-scale optimization problems \cite{ulybyshev2015LP} or consider advanced techniques, such as model predictive control \cite{Du2022MPC,elobaid2022MPC}, which significantly increase the computational burden of the application. To this end, some alternatives drop the optimality concern, pursuing instead formal asymptotic stability guarantees, typically through standard techniques anchored on Lyapunov theory, e.g. nonlinear backstepping \cite{Nunes2025backstepping} or feedback linearization \cite{Nazari2017LQR_Backstepping}.

An approach that provides a balance between optimality considerations, formal guarantees, and computational lightness is that of Linear Quadratic Regulator (LQR) theory. But while LQR-based station-keeping strategies already exist in literature \cite{Nazari2017LQR_Backstepping, ghorbani2013optimalcont}, these typically employ fixed weight matrices, tuned heuristically, that do not take the dynamics at play fully into account. As a result, the controller performance may be subpar, in spite of the optimization efforts.


To address the shortcomings identified in the literature, this work proposes a linear optimal control law for orbital station-keeping, derived from LQR theory, that leverages user-defined objectives with time-varying stability information of the target orbits, retrieved through Floquet theory. To ensure generality, the proposed solution is developed under the two most common dynamical models for spacecraft motion in cislunar space: the circular and elliptic R3BPs. We show that this solution minimizes a dynamically-informed cost functional and ensures local asymptotic stability when targeting any orbit within both models of the dynamics.

The remainder of this paper is organized as follows. In Section~\ref{sec:dyn}, the two R3BP formulations are presented and compared. Then, a brief introduction to Floquet theory and the stability of periodic orbits is provided in Section~\ref{sec:floquet_theory}. The design of the proposed solution and outline of its formal guarantees are approached in Section~\ref{sec:cont}. In Section~\ref{sec:results}, the solution is evaluated numerically over relevant application cases. Finally, Section~\ref{sec:conc} provides a discussion on the most meaningful conclusions and topics for future investigation.

\section{R3BP Dynamics}
\label{sec:dyn}

The R3BP provides a model for the dynamics of a small mass (e.g., a spacecraft) under the gravitational pull of two nearby massive bodies, named primaries. Given that the spacecraft's mass is negligible in comparison with those of the primaries, the latter constitute an isolated system that is known to have a limited set of solutions \cite{szebehely1967TheoryOfOrbits}. The case of Keplerian orbits around the shared barycenter is especially relevant, given that it provides an adequate approximation for the Earth-Moon relative motion. Two possibilities may thus be considered:
(i) the Circular Restricted Three-Body Problem (CR3BP), where the orbits of the primaries are approximated to circles, or (ii) the Elliptic Restricted Three-Body Problem (ER3BP), a more realistic, albeit more complex case, where elliptical orbits are considered instead.

In both versions of the R3BP, the spacecraft's EoM are established under a \textit{synodic} reference frame that rotates to match the primaries' angular speed. The $X$ direction aligns both primaries, $Z$ is directed along the angular speed vector, and $Y$ completes the right-handed coordinate system. Typical scaling brings the distance between primaries to unity and their orbital period to $2\pi$, resulting in unit mean angular speed. Moreover, their masses are brought to ${m_1=1-\mu}$ and ${m_2=\mu}$, with $\mu$ a system-specific constant (${\mu\approx0.01215}$, for the Earth-Moon). This fixes the positions of $m_1$ and $m_2$ to $(-\mu,0,0)$ and $(1-\mu,0,0)$, respectively.

\subsection{Circular Restricted Three-Body Problem}

Let ${\mathbf{r}^T=\begin{bmatrix}
		x & y & z
\end{bmatrix}}$ and ${\mathbf{v}^T=\begin{bmatrix}
		\dot{x} & \dot{y} & \dot{z}
\end{bmatrix}}$ denote the spacecraft position and velocity, respectively, under the synodic reference frame and proposed scaling, adopting the dot notation to represent time derivatives.
Letting $r_1$ and $r_2$ denote the (scalar) distance between the spacecraft and each primary allows one to write an effective CR3BP potential as $U = \frac{x^2+y^2}{2} + \frac{1-\mu}{r_1} + \frac{\mu}{r_2}$. The resulting spacecraft EoM are
\begin{equation}
	\label{eq:eom_matrix_form}
	\begin{aligned}
		\dot{\mathbf{r}} &= \mathbf{I_3}\mathbf{v},\\
		\dot{\mathbf{v}} &= \boldsymbol{\Omega}\mathbf{v} + \nabla U(\mathbf{r}),
	\end{aligned}
	\qquad \text{with} \qquad 
	\boldsymbol{\Omega}=\begin{bmatrix}
		0 & 2 & 0 \\ -2 & 0 & 0 \\ 0 & 0 & 0
	\end{bmatrix},
\end{equation}
where $\mathbf{I_3}$ is the $3\times3$ identity matrix, and $\nabla U(\mathbf{r})$ is the gradient of the effective potential, evaluated at $\mathbf{r}$, \cite{szebehely1967TheoryOfOrbits}.

\subsection{Elliptic Restricted Three-Body Problem}
\label{subsec:er3bp}

Since the angular speed and distance between primaries are not constant in the ER3BP, the reference frame must rotate at varying speed and pulsate to fix the primaries' location. We denote quantities in this frame by $\bar{(~\cdot~)}$. Moreover, the \textit{true anomaly}\footnote{The true anomaly corresponds to the angular position of the primaries along their orbits, relative to the major ellipse axis. The primaries are at their closest for $\nu=0$ (\textit{periapsis}), and at their farthest for $\nu=\pi$ (\textit{apoapsis}).} of the primaries is typically adopted as the independent variable, instead of time. Under these considerations, the equivalent ER3BP effective potential may be written as $\bar U = \frac{r}{p} \left( \frac{\bar x^2 + \bar y^2}{2} + \frac{1-\mu}{\bar r_1} + \frac{\mu}{\bar r_2} - \frac{1}{2} e \cos(\nu) \bar z^2 \right)$, where ${r=\frac{p}{1+e\cos(\nu)}}$ is the dimensionless distance between primaries, $e$ is the eccentricity of their orbits, and ${p := (1-e^2)}$ \cite{szebehely1967TheoryOfOrbits}. Adopting $(~\cdot~)'$ to denote derivatives with respect to $\nu$, the resulting EoM for the spacecraft are
\begin{equation}
	\label{eq:eom_matrix_form_er3bp}
	\begin{aligned}
		\bar{\mathbf{r}}' &= \mathbf{I_3}\bar{\mathbf{v}},\\
		\bar{\mathbf{v}}' &= \boldsymbol{\Omega}\bar{\mathbf{v}} + \nabla \bar{U}(\bar{\mathbf{r}},\nu),
	\end{aligned}
	\quad \text{with} \qquad
	\begin{aligned}
		\bar{\mathbf{r}}^T &=\begin{bmatrix}
			\bar{x} & \bar{y} & \bar{z}
		\end{bmatrix},\\
		\bar{\mathbf{v}}^T &=\begin{bmatrix}
			\bar{x}' & \bar{y}' & \bar{z}'
		\end{bmatrix}.
	\end{aligned}
\end{equation}
In contrast with \eqref{eq:eom_matrix_form}, the EoM \eqref{eq:eom_matrix_form_er3bp} are non-autonomous.

\subsection{Lagrange Points and Periodic Orbits}

The Lagrange points are five well-known equilibrium points of both EoM \eqref{eq:eom_matrix_form} and \eqref{eq:eom_matrix_form_er3bp}, and are typically ordered $L_1$ through $L_5$ by increasing effective potential \cite{szebehely1967TheoryOfOrbits}. Yet, rather than targeting the Lagrange points directly, the vast majority of missions track closed periodic orbits around them, since these offer ideal conditions for a wide variety of objectives. Most notably, the $L_1$ and $L_2$ orbital families have gained particular acclaim over recent years for being at the basis of the trajectories planned for NASA's Artemis missions \cite{creech2022Artemis}.


\section{Floquet Theory}
\label{sec:floquet_theory}

Linearization of the CR3BP EoM \eqref{eq:eom_matrix_form} about a nominal trajectory, denoted by ${\mathbf{x}^*}^T=\begin{bmatrix}
{\mathbf{r}^*}^T & {\mathbf{v}^*}^T
\end{bmatrix}$, yields the dynamics
\begin{equation}
	\label{eq:linearized_eom_cr3bp}
	\dot{\mathbf{z}} = \mathbf{A}(\mathbf{r}^*) \mathbf{z}, \qquad \text{with} \qquad \mathbf{A}(\mathbf{r}^*)=\begin{bmatrix}
		\mathbf{0}_3 & \mathbf{I}_3\\
		\mathbf{H}_U(\mathbf{r}^*) & \mathbf{\Omega}
	\end{bmatrix},
\end{equation}
where $\mathbf{0}_3$ is the $3\times3$ null matrix, $\mathbf{H}_U(\mathbf{r}^*)$ is the Hessian of $U$, evaluated at $\mathbf{r}^*:=\mathbf{r}^*(t)$, and $\mathbf{z}=\mathbf{x}-\mathbf{x}^*$ is the deviation from the target trajectory at each instant. In this work, we focus on tracking periodic orbits, such that \eqref{eq:linearized_eom_cr3bp} corresponds to a linear periodic system.
Floquet theory, which provides relevant insights on the stability of these systems, may thus be applied to infer on the stability of the target orbits themselves. 

To explore the application of Floquet theory, we rewrite \eqref{eq:linearized_eom_cr3bp} to make the independent variable explicit, i.e.
\begin{equation}
	\label{eq:generic_system}
	\dot{\mathbf{z}}(t) = \mathbf{A}(t) \mathbf{z}(t),
\end{equation}
noting that $\mathbf{A}(t+T)=\mathbf{A}(t)$, with $T$ the period of the target orbit.
Letting $\boldsymbol{\Phi}(t)$ denote the state-transition matrix (STM), with initial condition $\boldsymbol{\Phi}(0)=\mathbf{I}_n$, Floquet theory dictates that
\begin{equation}
	\label{eq:floquet_theorem}
	\boldsymbol{\Phi}(t) = \mathbf{P}(t)e^{t\mathbf{J}}\mathbf{P}^{-1}(0), ~\forall t \in \mathbb{R}^+,
\end{equation}
where $\mathbf{P}(t)$ and $\mathbf{J}$ are (possibly complex) $n\times n$ matrices -- the former $T$-periodic and non-singular for all $t$, and the latter constant \cite{wiesel1984ControlOfTimePeriodicSystems}. Since the STM is a solution of the system, substitution of \eqref{eq:floquet_theorem} in \eqref{eq:generic_system} yields the dynamics that govern $\mathbf{P}(t)$, i.e.
\begin{equation}
	\label{eq:floquet_P_dynamics}
	\dot{\mathbf{P}}(t) = \mathbf{A}(t) \mathbf{P}(t) - \mathbf{P}(t) \mathbf{J},
\end{equation}
after simplifications. Furthermore, since $\mathbf{P}(t)$ is $T$-periodic, taking $t=T$ in \eqref{eq:floquet_theorem} yields the \textit{monodromy matrix}, i.e.
\begin{equation}
	\label{eq:monodromy}
	\mathbf{M}:=\boldsymbol{\Phi}(T)=\mathbf{P}(T)e^{T\mathbf{J}}\mathbf{P}^{-1}(0) = \mathbf{P}(0)\boldsymbol{\Lambda}\mathbf{P}^{-1}(0),
\end{equation}
where $\boldsymbol{\Lambda}=e^{T\mathbf{J}}$ \cite{wiesel1983ModalControl}. Therefore, $\boldsymbol{\Lambda}$ collects the eigenvalues of $\mathbf{M}$, $\lambda_i,~i=\{1,\dots,n\}$, in Jordan form, and the respective eigenvectors are gathered in $\mathbf{P}(0)$, column-wise.

Floquet theory also establishes what is known as the \textit{modal transformation}, which significantly simplifies the dynamics at hand. This transformation takes the form ${\mathbf{z}(t)=\mathbf{P}(t)\boldsymbol{\eta}(t)}$, and may be differentiated and substituted in \eqref{eq:generic_system} to yield
\begin{equation*}
	\dot{\mathbf{P}}(t)\boldsymbol{\eta}(t) + \mathbf{P}(t)\dot{\boldsymbol{\eta}}(t) = \mathbf{A}(t)\mathbf{P}(t)\boldsymbol{\eta}(t).
\end{equation*}
Substitution of the dynamics of $\mathbf{P}(t)$, from \eqref{eq:floquet_P_dynamics}, leads to a linear time invariant counterpart of the original system,
\begin{equation*}
	\label{eq:floquet_LTI_sys}
	\dot{\boldsymbol{\eta}}(t) = \mathbf{J} \boldsymbol{\eta}(t).
\end{equation*}
The resulting coordinates, $\boldsymbol{\eta}$, represent the system's \textit{modes}, and $\mathbf{J}$ encodes the stability of each base direction \cite{wiesel1983ModalControl}. In fact, if $\boldsymbol{\Lambda}=\mathrm{diag}(\lambda_1,\dots,\lambda_n)$, then $\mathbf{J}=\mathrm{diag}(\sigma_1,\dots,\sigma_n)$, such that the $i$-th mode follows $\eta_i(t)=e^{t\sigma_i}\eta_i(0)$, where $\eta_i(0)$ is an initial condition and $\sigma_i=\frac{1}{T}\log(\lambda_i)$ is the corresponding Poincaré \textit{exponent} \cite{wiesel1984ControlOfTimePeriodicSystems}. The modal transformation thus exposes how the stable and unstable directions of the nominal orbit evolve with time, under the linear approximation, which is exploited in the design of the control law proposed in this work. A modal transformation in the ER3BP, relative to $\nu$, may be constructed in a similar manner.

\subsection{Remarks}
\label{subsec:floquet_remarks}

We remark that the existence of \textit{degenerate} modes may lead to linearly dependent directions in $\mathbf{P}(0)$ \cite{wiesel1984ControlOfTimePeriodicSystems}. In these cases, generalized eigenvalues may be employed to ensure $\mathbf{P}(0)$ is full-rank, such that \eqref{eq:floquet_theorem} and \eqref{eq:monodromy} still hold. Note also that retrieving $\mathbf{P}(t)$ requires the integration of \eqref{eq:floquet_P_dynamics}, from $\mathbf{P}(0)$. Given that, generally speaking, $\mathbf{P}(0)$ and $\mathbf{J}$ are complex matrices, this step may be numerically cumbersome. In addition, we will shortly discuss why a complex transformation matrix is undesirable for the purposes of developing the proposed control law. Fortunately, complex Poincaré exponents usually form conjugate pairs, such that $\mathbf{P}(0)$ and $\mathbf{J}$ may be written as real matrices by separating the exponents and eigenvectors into their real and imaginary parts \cite{wiesel1983ModalControl}. Exceptions evidence additional~dynamical symmetries, which may typically be regularized by redefining the modal transformation over a multiple of $T$ \cite{wiesel1994CanonicalFT}. Such cases are rare and outside the scope of this work, as is the case with the vast majority of the literature on this subject.

\section{Controller Design}
\label{sec:cont}

This work proposes a LQR control law for orbital station-keeping. Given the similarities between the CR3BP and ER3BP dynamics, we explore the derivation of this solution only under the former. Note that, in both versions of the R3BP, the linearized EoM are uniformly completely controllable and observable -- two necessary conditions which will not be shown in this work, but may be verified algebraically through a recursive matrix series, following \cite{batistaUCCUCO}.

Control is introduced to the linearized EoM \eqref{eq:generic_system} through the means of accelerations, $\mathbf{u}(t)\in\mathbb{R}^3$, yielding the system 
\begin{equation}
	\label{eq:linearised_eom_cr3bp_control}
	\dot{\mathbf{z}}(t) = \mathbf{A}(t) \mathbf{z}(t) + \mathbf{B} \mathbf{u}(t), \quad \text{where} \quad \mathbf{B}^T = \begin{bmatrix}
		\mathbf{0_3} & \mathbf{I_3}
	\end{bmatrix}.
\end{equation}

Then, the LQR problem may be established through the definition of a typical cost function. In \cite{ghorbani2013optimalcont}, a finite-horizon version is pursued, leading to a Riccati differential equation (RDE) that may be solved directly but requires the imposition of a boundary condition, which represents an additional parameter to be selected. In contrast, the work in \cite{Nazari2017LQR_Backstepping} considers an infinite-horizon alternative that discretizes the resulting RDE by solving an algebraic Riccati equation at each instant. This, however, fails to approximate a true solution of the original RDE. Moreover, both approaches employ constant weight matrices in the cost functions, tuned heuristically, that are insensitive to the time-varying (periodic) dynamics at play and the stability properties of the target orbits.

To address the shortcomings identified in the literature, we establish an infinite-horizon LQR problem through a cost function in the form of
\begin{equation}
	\label{eq:lqr_cost_function}
	\mathcal{J} = \int_{0}^{\infty}[\mathbf{z}^T(t)\mathbf{Q}(t)\mathbf{z}(t) + \mathbf{u}^T(t)\mathbf{R}(t)\mathbf{u}(t)]~ dt,
\end{equation}
where $\mathbf{R}(t)=\mathbf{R}_0=\alpha \mathbf{I}_3, ~\alpha>0,$ is made constant but $\mathbf{Q}(t)=\mathbf{Q}_0 + \mathbf{P}^{-T}(t)\mathbf{Q}_p\mathbf{P}(t)$ is periodic, exploiting the modal transformation, $\mathbf{P}(t)$, from \eqref{eq:floquet_P_dynamics}. In particular, we take
\begin{equation*}
	\mathbf{Q}_0 = \begin{bmatrix}
		\beta_r \mathbf{I}_3 & \mathbf{0}_3\\
		\mathbf{0}_3 & \beta_v \mathbf{I}_3
	\end{bmatrix}
	\qquad \text{and} \qquad
	\mathbf{Q}_p = \mathrm{diag}(\gamma_1,\dots,\gamma_6),
\end{equation*}
where $\beta_r$, $\beta_v$, and $\gamma_i,~i=\{1,\cdots,6\}$ are positive constants. By establishing $\mathbf{Q}$ with these two contributions, it is possible to weight not only the physical position and velocity deviations, but also each linear orbital mode separately. In this sense, the resulting cost function is sensible about the dynamics, making it possible to prioritize the elimination of deviations along the most unstable directions of the target orbits. Evidently, ensuring that $\mathbf{P}(t)$ is a real matrix bestows the LQR weights with a physical interpretation -- hence the relevance of the discussion in Section~\ref{subsec:floquet_remarks}.

As per \cite{kalman1960contributions}, minimization of \eqref{eq:lqr_cost_function} requires solving the RDE
\begin{equation}
	\label{eq:prde}
	\begin{aligned}
		\dot{\mathbf{S}}(t) = &-\mathbf{S}(t)\mathbf{A}(t) - \mathbf{A}^T(t)\mathbf{S}(t) -\mathbf{Q}(t)\\
		&+ \mathbf{S}(t)\mathbf{B}(t)\mathbf{R}^{-1}(t) \mathbf{B}^T(t)\mathbf{S}(t),
	\end{aligned}
\end{equation}
which, in this case, is $T$-periodic.
Under the aforementioned controllability and observability guarantees, and since $\mathbf{Q}(t)$ and $\mathbf{R}(t)$ are bounded and positive-definite, \eqref{eq:prde} has a unique positive-definite periodic solution, $\mathbf{S}^\star(t)$, to which all other solutions converge backwards in time \cite[Theorem 6.6]{bittanti1991PRDE}. An approximation of $\mathbf{S}^\star(t)$ may thus be obtained by successively solving \eqref{eq:prde} backwards, considering a time-window equal to (a multiple of) the nominal orbit period. To this end, we take $\mathbf{S}_k(T)=\mathbf{S}_{k-1}(0)$ as the boundary condition of the $k$-th iteration, and repeat the process until $\mathbf{S}_k(T)\approx\mathbf{S}_k(0)$. Given the convergence guarantees, the initial positive-definite matrix $\mathbf{S}_0(T)$ may be chosen freely. After retrieving $\mathbf{S}^\star(t)$, the linear optimal control law may be defined as follows.
\begin{proposition}
	\label{prop:control_law_LQR_CR3BP}
	Given the linearized CR3BP EoM \eqref{eq:linearised_eom_cr3bp_control}, the linear control law
	\begin{equation}
		\label{eq:LQR_control_law_CR3BP}
		\mathbf{u} = -\mathbf{K}(t)\mathbf{z}, \qquad \text{with} \qquad \mathbf{K}(t) =\mathbf{R}_0^{-1}\mathbf{B}^T\mathbf{S}^\star(t),
	\end{equation}
	where $\mathbf{S}^\star(t)$ is the unique positive-definite periodic solution of \eqref{eq:prde}, guarantees asymptotic stability of $\mathbf{z}=\mathbf{0}$. Moreover, it ensures that the controlled response of the linearized system minimizes the cost function \eqref{eq:lqr_cost_function}.
\end{proposition}
\begin{proof}
	A general proof for both statements is provided in \cite[Theorem 6.7, Theorem 6.10]{kalman1960contributions}, and may be particularized to the specific case under analysis by following the steps and assumptions discussed in this section.
\end{proof}

Note that $\mathbf{P}(t)$, $\mathbf{S}^\star(t)$, and $\mathbf{K}(t)$ need to be computed only once, ideally before mission deployment. Through the interpolation of $\mathbf{K}(t)$, an on-board implementation of the proposed control law can thus be very computationally light.

\section{Numerical Results}
\label{sec:results}

We evaluate the station-keeping capabilities of the proposed solution by considering small perturbations to target orbits within the CR3BP and the ER3BP, denoted C1 and E1, respectively. Both orbits belong to the $L_2$ Halo family and are closely linked in terms of dynamics and shape, as explored in \cite{Nunes2025backstepping}. The respective orbital period and Poincaré exponents are provided in Table~\ref{tab:test_cases_stab}, denoting the imaginary unit by $j$. In both cases, a single unstable mode exists. Its magnitude reveals that both orbits are highly unstable and hence ideal candidates for the tests performed in this work.

\begin{table}[!h]
	\caption{Stability properties and period of the nominal orbits.}
	\label{tab:test_cases_stab}
	\begin{center}
		\begin{tabular}{c c c}
			\hline
			Orbit& Poincaré Exponents & Orbital Period\\ \hline
			C1 & $(\pm 1.607,~ \pm 0.572j,~ 0,~ 0)$ & $\sim13.9$ days\\
			E1 & $(\pm 1.609,~ \pm 0.430j,~ \pm 0.004j)$ & $\sim27.9$ days\\ \hline
		\end{tabular}
	\end{center}
\end{table}

Given that we are trying to emulate an electric propulsion system, the inclusion of an actuation dead-band is pondered to \textit{(i)} limit the amount of time that the actuators are active, and \textit{(ii)} constrain the magnitude of the maneuvers according to an established minimum. To this end, control is turned off once $\norm{\mathbf{u}}<u_\text{min}$, where $u_\text{min}>0$ is a constant that depends on the particular actuator system employed. Since small deviations about the nominal orbit are typically acceptable, control action is made active again only when $\norm{\mathbf{u}}>u_\text{min}$ and $\norm{\mathbf{z}}>z_\text{th}$, where $z_\text{th}>0$ is a mission-dependent threshold. In this work, the values $u_\text{sat}=1\times10^{-7}~\text{m s}^{-2}$ ($10~\mathrm{\mu N}$, for a $100~\text{kg}$ spacecraft) and $z_\text{th}=100~\text{km}$ were selected.

We consider an initial perturbation with a non-dimensional magnitude of $10^{-7}$ along the most unstable direction of the orbit, i.e. that of mode $\eta_1$. This corresponds to a deviation of around $40~\text{m}$ or $0.1~\text{mm s}^{-1}$ from the initial target state. Since the orbits are highly unstable and the maneuvers are lower bounded in magnitude, the initial perturbation will never be fully eliminated. It is thus expected that, as a result of entering the control dead-band, the deviation will naturally increase until the next maneuver. Control action then reduces the deviation from the target orbit, and the cycle repeats.

Our particular analysis will keep $u_\text{min}$ and $z_\text{th}$ fixed, in order to study the effects of including dynamical information in the design of the linear optimal control law, through $\mathbf{Q}_p$ in \eqref{eq:prde}. These effects are quantified through two benchmarks,
\begin{equation*}
	\mathcal{E}_{v}= \int_{0}^{t_\text{sim}} \norm{\mathbf{u}(t)} dt \qquad \text{and} \qquad \tau_\text{active} = \frac{t_\text{active}}{t_\text{sim}},
\end{equation*}
where $t_\text{sim}$ is the simulation time span, and $t_\text{active}$ is the total time that control is active. In this sense, $\mathcal{E}_v$ measures the total velocity variation enacted by the actuators and $\tau_\text{active}$ their relative activity during the simulation.

The simulations are run over $10$ revolutions about orbit C1 and $5$ revolutions about orbit E1. Numerical integration of the modal transformation, periodic RDE, and EoM is performed through an adaptive-step Runge-Kutta method of order $8(9)$, with relative and absolute error tolerances set to $10^{-12}$.

\subsection{Application in the CR3BP}
\label{subsec:app_CR3BP}

Since constant-weight LQR solutions for station-keeping have already been analyzed in the literature \cite{Nazari2017LQR_Backstepping, ghorbani2013optimalcont}, we focus on isolating the effect of weighting the linear orbital modes, particularly the unstable mode. To this end, a preliminary heuristic study was carried out to find the weights $(\beta_r,\beta_v,\alpha)$ that minimize $\mathcal{E}_v$, for $\gamma_i=0,~\forall i$. To exclude limit cases, the analysis was restricted to choices that ensure ${\text{max}\norm{\mathbf{z}}<3z_\text{th}}$. For the CR3BP dynamics, this led to $(\beta_r,\beta_v,\alpha)=(2,1,3)$. Any other choice was found to break the deviation constraint or increase $\mathcal{E}_v$. In particular, we note that employing more control action in an uninformed manner, such as reducing $\alpha$ or increasing $\beta_r$ and/or $\beta_v$, deteriorates the benchmark value.

Then, we successively increase $\gamma_1$ to boost the preponderance of the unstable mode in the cost function \eqref{eq:lqr_cost_function}. The remaining orbital modes are not weighted due to their favorable stability properties. Moreover, $\beta_r$, $\beta_v$, and $\alpha$ are kept constant to isolate the effect of $\gamma_1$. The resulting benchmarks stemming from this study are provided in Table~\ref{tab:CR3BP_benchmarks}, allowing for a direct comparison with the results of the uninformed, constant-weight LQR (for $\gamma_1=0$). The maximum magnitude of the deviation and control command are also provided. 

\begin{table}[!h]
	\caption{Benchmarks of the CR3BP test cases, for $(\beta_r,\beta_v,\alpha)=(2,1,3)$.}
	\label{tab:CR3BP_benchmarks}
	\begin{center}
		\begin{tabular}{c c c c c c c}
			\hline
			 $\gamma_1$ & $0$ & $10$ & $50$ & $100$ & $200$ & \\ \hline
			$\mathcal{E}_v$ & $3.258$ & $3.462$& $2.619$ & $2.355$ & $2.628$  & $[\text{m s}^{-1}]$\\
			$\tau_\text{active}$& $54.2$ & $41.4$ & $36.0$ & $31.2$ &  $48.3$ & $\%$ \\
			$\max\norm{\mathbf{z}}$ & $2.887$ & $1.741$ & $2.090$ & $2.424$ & $2.531$ & $(\times z_\text{th})$\\ 
			$\max\norm{\mathbf{u}}$ & $3.726$ & $4.679$ & $5.682$ & $6.252$ & $6.535$ & $[\mu\text{m s}^{-2}]$\\ 
			 \hline
		\end{tabular}
	\end{center}
\end{table}

From the analysis of Table~\ref{tab:CR3BP_benchmarks}, one concludes that a prudent selection of $\gamma_1$ allows for a decrease in $\mathcal{E}_v$ of up to $ 30\%$. Note that this occurs in spite of having previously established that trying to induce more control action in an uninformed manner negatively affects this benchmark. Hence, we guarantee that the benefits observed are a direct consequence of including dynamical information in the design of the control law. These benefits are equally felt in terms of controller activity, which may also be reduced significantly through an appropriate choice of $\gamma_1$. Moreover, they appear to bear no negative effect on the adequacy of the controller at tracking the target trajectory, given how $\max\norm{\mathbf{z}}$ is kept below the value resulting from the application of the base LQR law. 
However, the improvements are less pronounced and possibly reversed for very large choices of $\gamma_1$. Furthermore, such choices lead to increasingly more demanding actuation commands, as evidenced by $\max{\norm{\mathbf{u}}}$, which would make the proposed solution unfit for the practical application under study. To this end, we take $\gamma_1=100$ as the optimal choice.

To better evaluate the differences between the uninformed LQR law and an appropriately tuned counterpart (namely, with $\gamma_1=100$), consider the system responses in Fig.~\ref{fig:CR3BP_resp}. The magnitude of the position and velocity deviations, unstable mode, and control command are provided. Apart from the non-dimensional mode, all quantities are re-scaled into physical units through the R3BP characteristic scales detailed in Section~\ref{sec:dyn}. An activity chart, at the bottom, evidences the periods of time where control action is employed.

\begin{figure}[t]
	\centering
	\includegraphics[width=0.9\linewidth, trim={0cm 2cm 2cm 0cm}, clip]{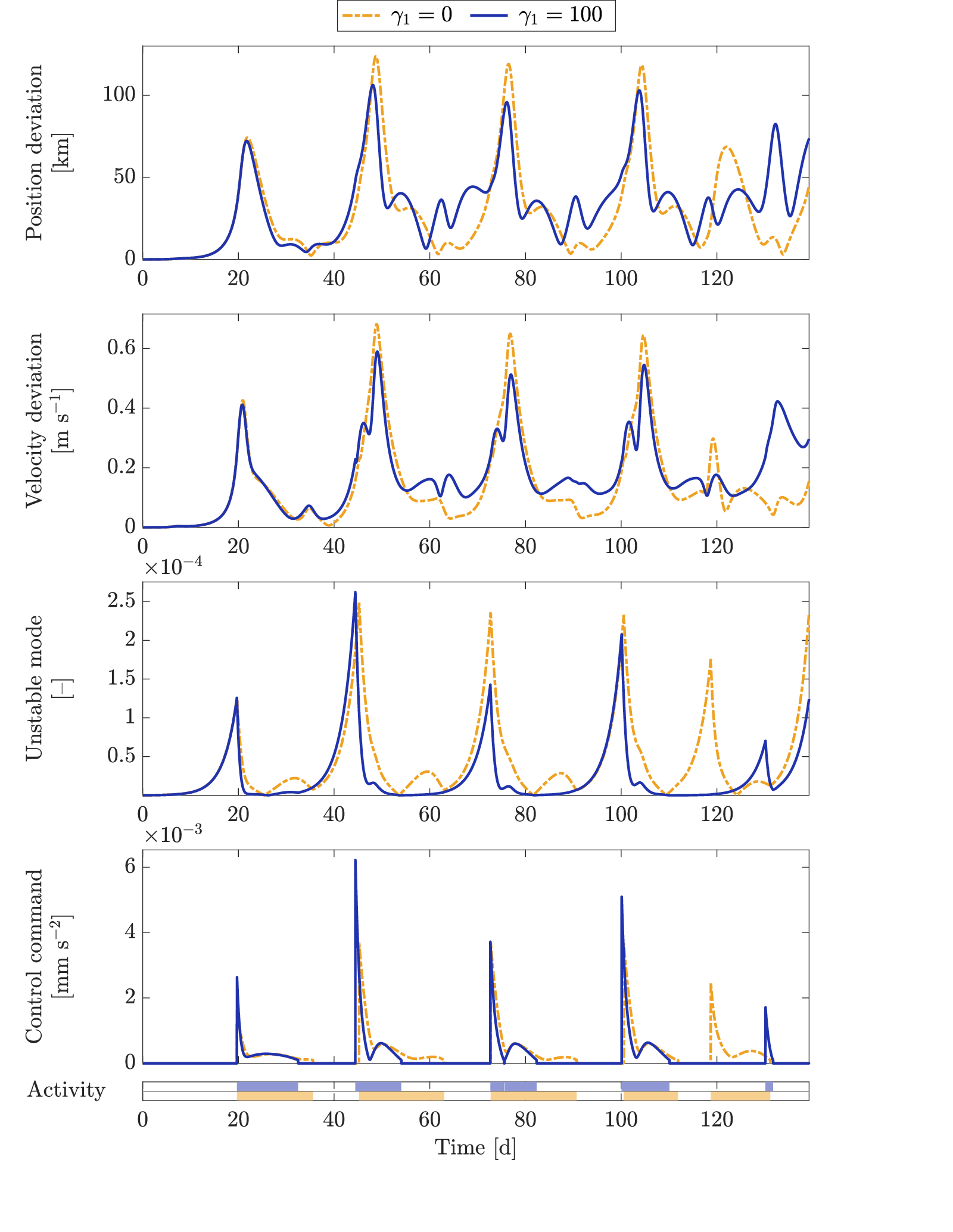}
	\caption{System responses under the CR3BP dynamics and the proposed LQR, with $\beta_r=2$, $\beta_v=1$, $\alpha=3$, and $\gamma_{i\neq1}=0$.  Two choices for the $\gamma_1$ weight are presented to evidence the difference between an uninformed LQR ($\gamma_1=0$) and a dynamically informed counterpart ($\gamma_1=100$).}
	\label{fig:CR3BP_resp}
\end{figure}

From the analysis of Fig.~\ref{fig:CR3BP_resp}, we confirm that the inclusion of dynamical information in the controller design can be made without a sacrifice to its tracking capabilities. This is evidenced by the physical deviations, which achieve peaks of similar magnitude to the uninformed case, after periods of control inactivity. The option with $\gamma_1=100$ is however much better suited at driving deviations along the unstable direction to zero, when maneuvers are employed. This leads to a two-fold improvement: \textit{(i)} the control solution is faster at eliminating the deviations, which grow predominantly along the unstable direction, and \textit{(ii)} by further reducing $\eta_1$ over each actuation window, the spacecraft remains in proximity to the target orbit for a longer time when control is inactive. As evidenced by the activity chart, these effects meaningfully reduce total maneuver time, in line with the benchmarks in Table~\ref{tab:CR3BP_benchmarks}. We draw however attention to the increase in magnitude in some of the control commands, which explains the counteracting effects verified for very large choices of $\gamma_1$.

To further validate the adequacy of the proposed solution, the controlled spacecraft trajectory is plotted in Fig.~\ref{fig:CR3BP_orbit}, under the same weights and $\gamma_1=100$. In contrast, the trajectory of an uncontrolled spacecraft is also provided, highlighting the need for station-keeping. A closeup $X$--$Y$ view of the region near the initial conditions, marked by a cross, evidences the tracking accuracy of the strategy employed and the departure of the uncontrolled spacecraft after a single revolution.

\begin{figure}[t]
	\centering
	\includegraphics[width=0.89\linewidth, trim={1cm 1.2cm 1cm 0cm}, clip]{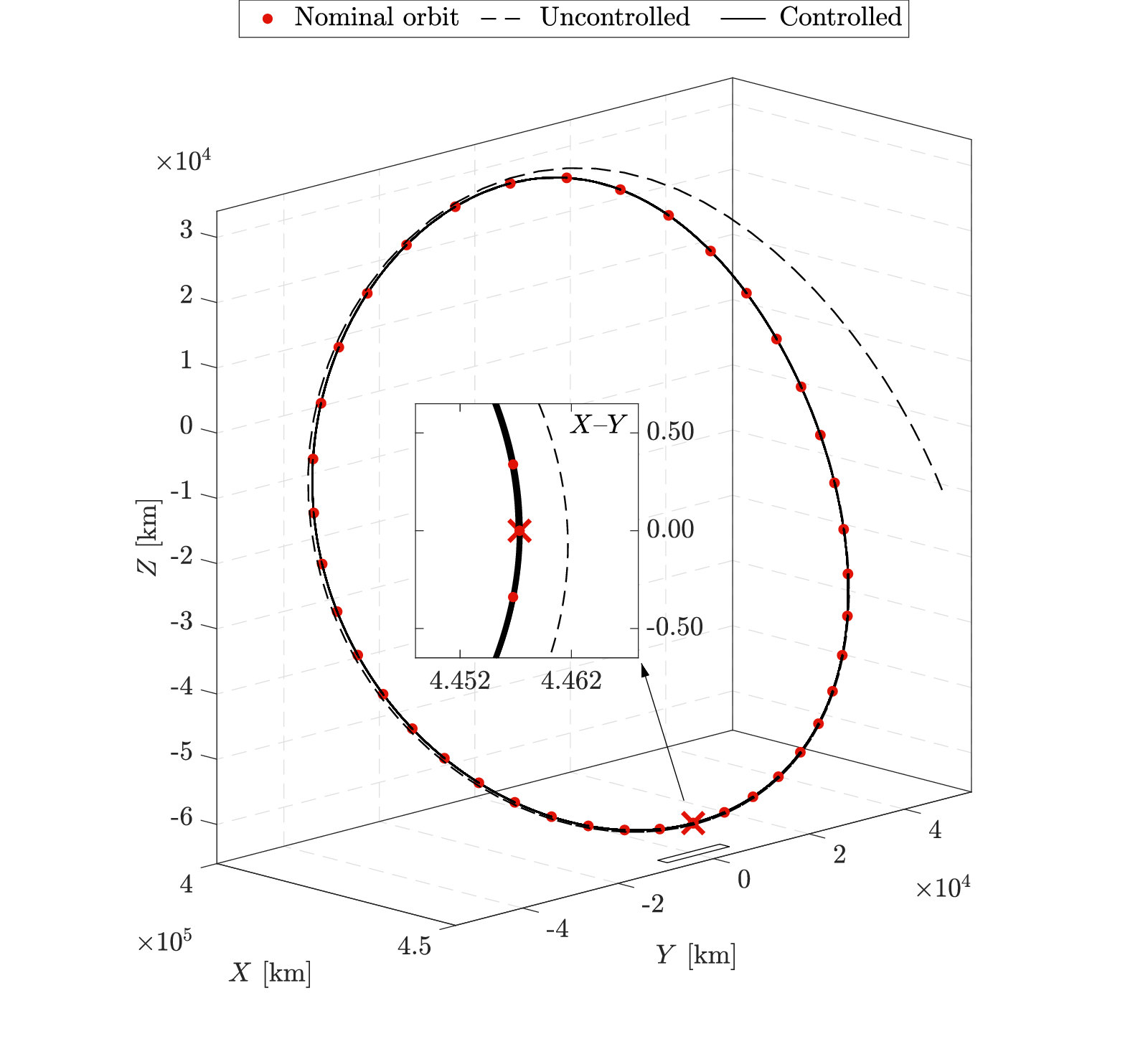}
	\caption{Trajectories in the CR3BP with $\beta_r =2$, $\beta_v=1$, ${\alpha=3}$, ${\gamma_1=100}$, and $\gamma_{i\neq1}=0$. A closeup $X$--$Y$ view near the initial conditions evidences the departure of an uncontrolled spacecraft after one revolution.}
	\label{fig:CR3BP_orbit}
\end{figure}

\subsection{Application in the ER3BP}

To study the application of the proposed station-keeping solution in the ER3BP, we proceed in a similar manner to the CR3BP case. Note, however, that an extra step is necessary to convert from the independent variable, $\nu$, to time. As detailed in \cite{Nunes2025backstepping}, this is done through Kepler's equation, assuming the primaries start at periapsis, i.e. $\nu\vert_{t=0}=0$. The optimization of $\mathcal{E}_v$ under similar considerations to those of Section~\ref{subsec:app_CR3BP} leads once more to the set of weights $(\beta_r,\beta_v,\alpha)=(2,1,3)$ -- which is unsurprising given the similarity between orbits C1 and E1, as previously evidenced in Table~\ref{tab:test_cases_stab}.

By increasing $\gamma_1$, one achieves the benchmarks in Table~\ref{tab:ER3BP_benchmarks}. We conclude that the increase in dynamical complexity of the ER3BP does not significantly change the conclusions previously drawn out for the CR3BP. Most notably, we once more observe that a careful selection of the $\gamma_1$ weight leads to benefits in control effort and total activity, without sacrificing the remaining performance measures. Similarly, increasing the weight past $\gamma_1>100$ proves to be unfruitful.

\begin{table}[!h]
	\caption{Benchmarks of the ER3BP test cases, for $(\beta_r,\beta_v,\alpha)=(2,1,3)$.}
	\label{tab:ER3BP_benchmarks}
	\begin{center}
		\begin{tabular}{c c c c c c c}
			\hline
			$\gamma_1$ & $0$ & $10$ & $50$ & $100$ & $200$ & \\ \hline
			$\mathcal{E}_v$ & $3.224$ & $5.308$& $2.815$ & $2.597$ & $2.042$  & $[\text{m s}^{-1}]$\\
			$\tau_\text{active}$& $41.9$ & $46.8$ & $27.5$ & $25.9$ &  $29.3$ & $\%$ \\
			$\max\norm{\mathbf{z}}$ & $2.866$ & $1.813$ & $2.207$ & $2.480$ & $2.527$ & $(\times z_\text{th})$\\
			$\max\norm{\mathbf{u}}$ & $5.159$ & $4.087$ & $4.882$ & $5.280$ & $5.471$ & $[\mu\text{m s}^{-2}]$\\ \hline
		\end{tabular}
	\end{center}
\end{table}

In Fig.~\ref{fig:ER3BP_resp}, the system responses under the LQR with $\gamma_1=0$ and $\gamma_1=100$ are provided. Analogously to the CR3BP case, the inclusion of this weight is found to result in a better tuned response that strongly opposes deviations along the direction of the most unstable mode, prolonging the windows of control downtime when entering the dead-band. Once again, these improvements are met without putting into question the tracking adequacy of the solution, as measured by the evolution of the position and velocity deviations.

\begin{figure}[t]
	\centering
	\includegraphics[width=0.9\linewidth, trim={0cm 2.1cm 2cm 0cm}, clip]{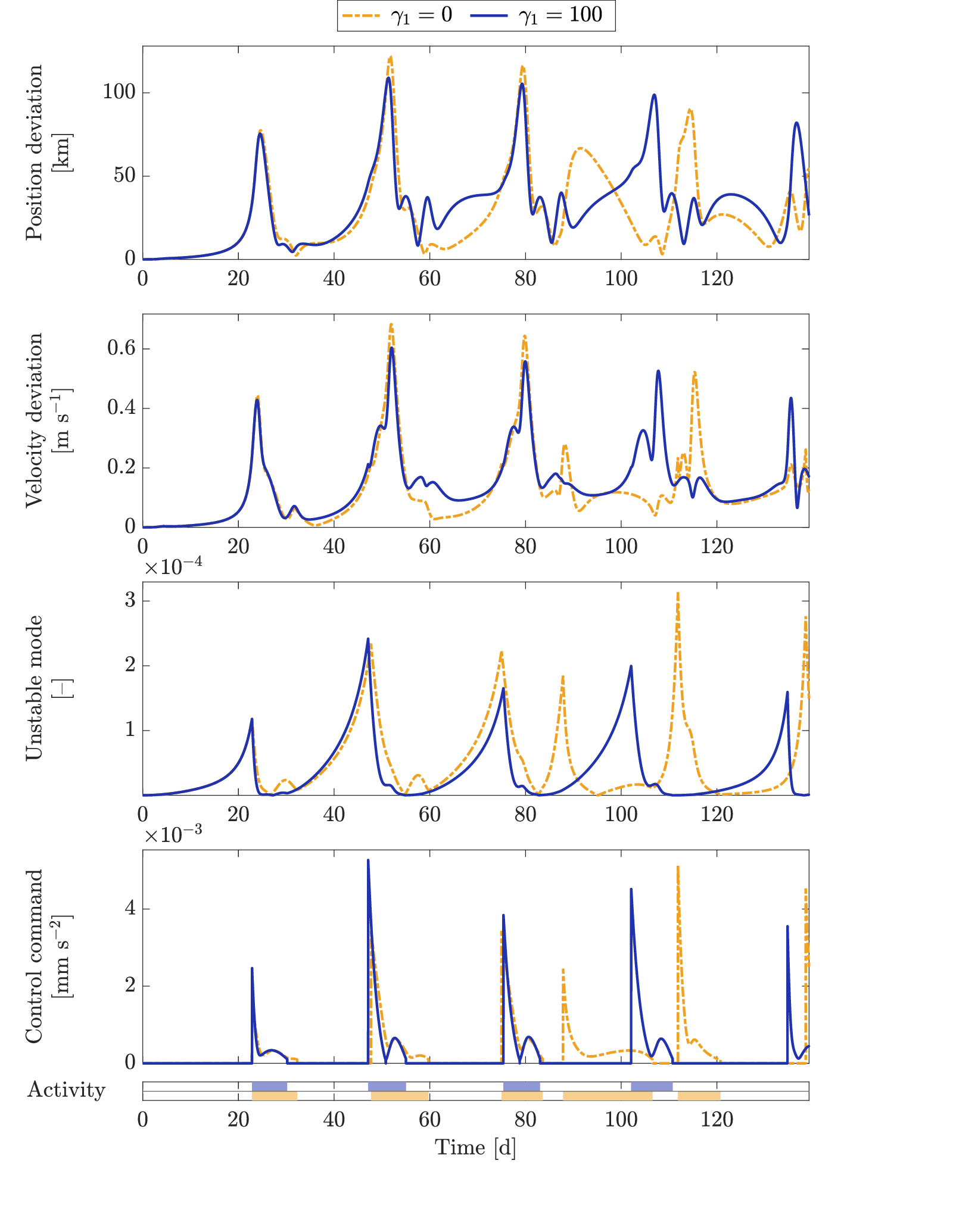}
	\caption{System responses under the ER3BP dynamics and the proposed LQR, with $\beta_r=2$, $\beta_v=1$, $\alpha=3$, and $\gamma_{i\neq1}=0$.  Two choices for the $\gamma_1$ weight are presented to evidence the difference between an uninformed LQR ($\gamma_1=0$) and a dynamically informed counterpart ($\gamma_1=100$).}
	\label{fig:ER3BP_resp}
\end{figure}

The effectiveness of the station-keeping approach under the previous weights and $\gamma_1=100$ is visually confirmed in Fig.~\ref{fig:ER3BP_orbit}. To this end, one may state that the control solution is equally adequate in both representations of the R3BP.

\begin{figure}[t]
	\centering
	\includegraphics[width=0.89\linewidth, trim={1cm 1.2cm 1cm 0cm}, clip]{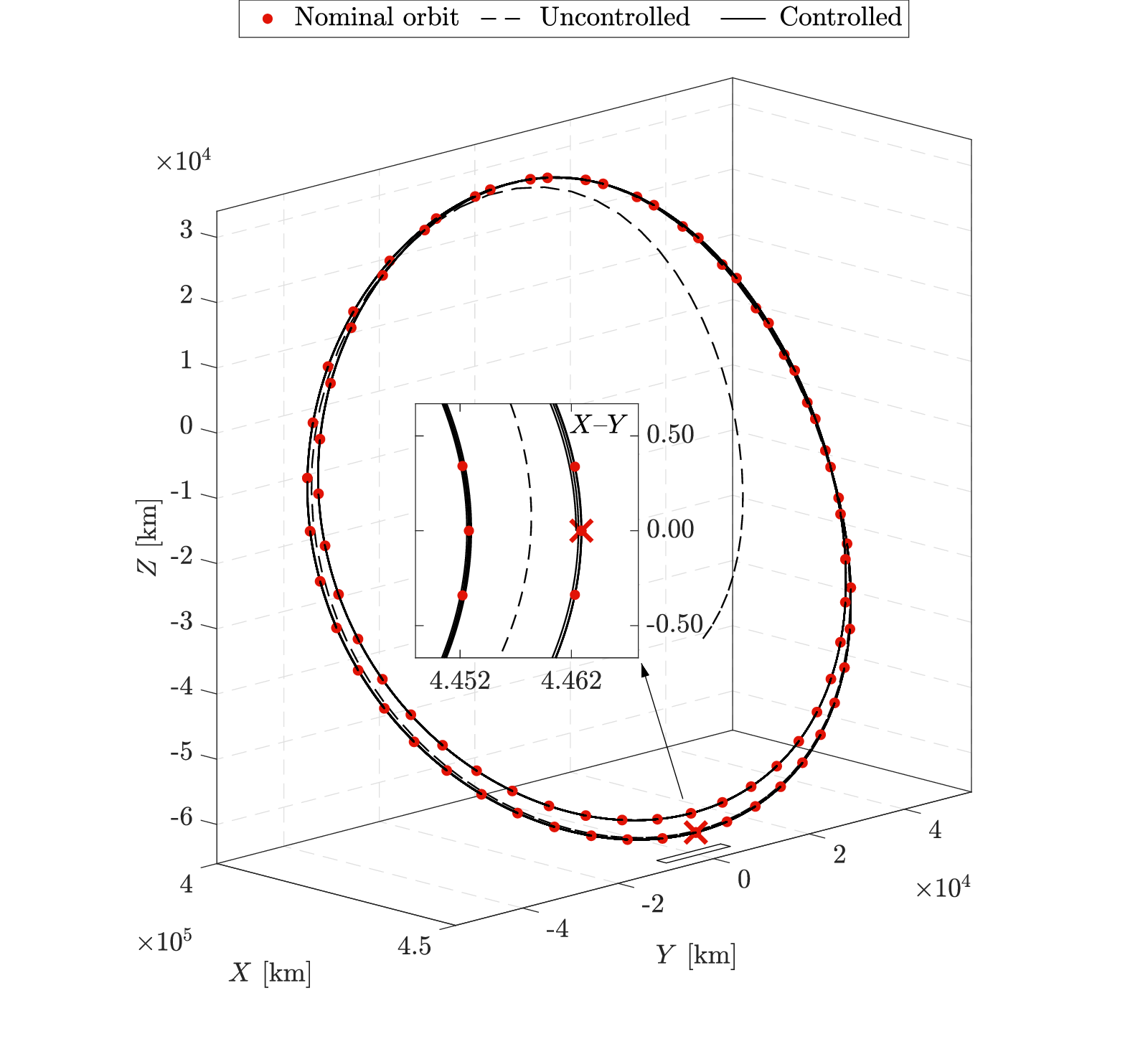}
	\caption{Trajectories in the ER3BP with $\beta_r =2$, $\beta_v=1$, ${\alpha=3}$, ${\gamma_1=100}$, and $\gamma_{i\neq1}=0$. A closeup $X$--$Y$ view near the initial conditions evidences the departure of an uncontrolled spacecraft after two revolutions.}
	\label{fig:ER3BP_orbit}
\end{figure}

\section{Concluding Remarks}
\label{sec:conc}
This work presents a linear optimal control law for continuous, low-thrust orbital station-keeping in the CR3BP and ER3BP. In contrast to existing constant-weight approaches, a cost function with periodic weights is formulated by combining LQR and Floquet theories to better exploit the underlying dynamics of the target orbits. The established cost function is minimized by solving a periodic RDE in an iterative manner. The resulting control law is shown to guarantee asymptotic stability of the linearized EoM. Through numerical simulations, the station-keeping capabilities of the control law are verified, and the effect of including dynamical information in the cost function is isolated and studied. We find that, when appropriately tuned, this consideration significantly improves control efforts and, if a dead-band is employed, also acts to reduce total controller activity.

Future work should assess the effect of external perturbations, e.g. solar radiation pressure or the gravitational influence of other celestial bodies. Moreover, extension to high-fidelity dynamical models accounting for realistic primary trajectories, retrieved from ephemeris data, should be studied.

%



\bibliographystyle{IEEEtran}
\bibliography{references}

\end{document}